\documentclass{llncs}

\usepackage{tgpagella,eulervm}

\usepackage{amsmath}
\usepackage{amssymb}

\usepackage{amsthm}

\let\es\varnothing

\DeclareMathOperator{\AT}{\mathsf{AT}}
\DeclareMathOperator{\NP}{\mathsf{NP}}
\DeclareMathOperator{\Max}{\mathsf{Max}}
\DeclareMathOperator{\Min}{\mathsf{Min}}
\DeclareMathOperator{\Opt}{\mathsf{Opt}}
\DeclareMathOperator{\ess}{\mathsf{ess}}

\newtheorem*{notation}{Notation}

\begin{document}

\title{{\sc Flood-it on $\AT$-Free Graphs}\\
{\small A body of clay, a mind full of play, a moment's life -- that's me.}\thanks{Harivansh Rai Bachchan}}

\author{Wing-Kai Hon\inst{1} 
\and 
Ton Kloks\inst{2}
\and 
Fu-Hong Liu\inst{1}
\and 
Hsiang-Hsuan Liu\inst{1}
\and
Hung-Lung Wang\inst{3}}

\institute{National Tsing Hua University\\
Hsinchu, 
Taiwan\\
\email{(wkhon,fhliu,hhliu)@cs.nthu.edu.tw}
\and
National Taipei University of Business\\
Taipei, 
Taiwan\\
\email{hlwang@ntub.edu.tw}}

\maketitle              

\begin{abstract}
Solitaire {\sc Flood-it}, or {\sc Honey-Bee}, is a game played 
on a colored graph. The player resides in a source vertex. Originally 
his territory is the maximal connected, monochromatic subgraph that 
contains the source. A move consists of 
calling a color. This conquers all the 
nodes of the graph that can be reached by a monochromatic path of that color   
from the current territory of the player. It is the aim of the player 
to add all vertices to his territory in a minimal number of moves. 
We show that the minimal 
number of moves can be computed in polynomial time when the 
game is played on $\AT$-free graphs. 
\end{abstract}

\section{Introduction}
As Oscar Wilde already observed, `Life is far too important a thing ever to 
talk seriously about,' so,  today,  let's play! 

\bigskip 

{\sc Flood-it} is a popular game which can be played solitaire or 
together with other people or with machines.  We define the solitaire 
game in terms of a graph 
as follows. The input is a graph $G$ of which the vertices are colored. 
Let 
\[c:V(G) \rightarrow C \quad \text{and}\quad C=\{\;1,\;\dots,\;k\;\}\]  
denote the coloring of the vertices of $G$ with colors from a set $C$ of $k$ colors. 
Originally, the player ocupies one vertex, say $x_0$, 
and his `territory' consists of $x_0$ plus all the vertices that can be 
reached from $x_0$ by a monochromatic path of 
vertices of color $c(x_0)$.  A move consists of the player's calling 
of a color, say $i$, which, one may assume, 
is not the current color of his territory. 
The result 
of the move is that the territory of the player is recolored with the 
color $i$ and that those vertices are added to it that can be reached from 
$x_0$ by a monochromatic path of color $i$. 
The aim of the player is to increase 
his territory to $V(G)$ in as few moves 
as possible. 

\bigskip 

In another version, called {\sc Free-Flood-it}, the player may 
grow his territory from different starting points at any move. 
In a 2-player version, two players play against each other. The player 
who grows his territory to at least half of the vertices wins the game. 

\bigskip 

In this paper we confine ourselves to the analysis 
of the solitaire game described above.  In the remainder of this 
introduction we briefly mention some of the known complexity results. 
The problem is polynomial for paths, cycles, and cocomparability 
graphs. It is 
$\NP$-complete for splitgraphs and on trees, even when the 
number of colors is restricted to 3. Furthermore, the problem 
remains $\NP$-complete on $3 \times n$ boards, even when the 
number of colors is only 4. 
The free version is $\NP$-complete 
on trees and even on caterpillars.  When parameterized by the number 
of colors, the free problem becomes tractable again for interval 
graphs and for splitgraphs. 
For $2 \times n$ boards the free problem is 
fixed-parameter tractable, when parameterized by the number of colors. 
The problem remains $\NP$-complete on such boards when the number 
of colors is unbounded. 
Notice that 
the solitaire game is trivial in case there 
are only two colors. 

\section{Preliminaries on $\AT$-free graphs}

Asteroidal sets were introduced by Walter in 1978 to characterize 
certain subclasses of chordal graphs. They were rediscovered and 
put to use in~\cite{kn:broersma,kn:kloks2}. 
The elements of asteroidal sets of cardinality 3 are called asteroidal 
triples. They were introduced by Lekkerkerker and 
Boland in their pioneering paper on the 
characterization of interval graphs~\cite{kn:lekkerkerker}. 
To be precise, the authors of the underlying work 
characterize interval graphs as those chordal 
graphs without asteroidal triples. For a somewhat different proof, 
and an extension to the infinite, we refer to~\cite{kn:halin}. 
See also~\cite{kn:kloks} for another description 
of the proof and for definitions of basic concepts that we assume here 
familiarity with. 

\begin{definition}
In a graph, an asteroidal triple is an independent set of three 
vertices such that every pair of them is connected by a path 
that avoids the closed neighborhood of the third. 
\end{definition}

In an attempt to find an asteroidal triple in a graph one can examine every 
triple; remove, in turn, the closed neighborhood of a triple's element, 
and check whether 
the remaining pair is contained in one component of the truncated graph. 
To check if a graph is $\AT$-free this is, at the moment, basically 
the best known method, since it can be shown that finding 
an asteroidal triple in a graph is at least as hard as finding a 
triangle in a graph (see, eg,~\cite{kn:kohler}). 

\bigskip 

Graphs without asteroidal triples, that is, $\AT$-free graphs, 
generalize cocomparability graphs in a natural way.   
$\AT$-free graphs properly 
contain cocomparability graphs and these, in turn,  
contain valued 
classes such as interval graphs and permutation graph. 
Notice however, that the class of 
$\AT$-free graphs differs from the smaller ones by the 
fact that its elements are, on the whole, not perfect. 
Recall that a graph is the complement of a comparability graph 
if it has an intersection model in which each vertex $x$ is represented 
by a continuous function $f_x: [0,1] \rightarrow \mathbb{R}$. Two 
vertices are adjacent in the graph if their functions 
intersect~\cite{kn:golumbic,kn:kloks}.   To see that they are 
indeed $\AT$-free, observe that for any three, pairwise non-intersecting,  
continuous functions, one must lie between the other two. Then every path 
that runs between the outer pair, must have a vertex in the  
closed neighborhood of the one in the middle. 

\bigskip 

Another property satisfied by cocomparability graphs is that the class is 
closed under edge contractions. 

\begin{definition}
Let $G$ be a graph and let $\{x,y\} \in E(G)$. A contraction 
of the edge $\{x,y\}$ replaces the pair by a single, new vertex. 
The neighborhood of the new vertex is the union of the neighborhoods 
of the endpoints of the edge,  
with the omission of the deleted endpoints $x$ and $y$: 
\[N(x) \cup N(y) \setminus \{\;x,\;y\;\}.\]
\end{definition}

Edge contractions play an important role in the theory of graph minors. 
A minor of a graph $G$ is a graph obtained from $G$ by a series 
of edge- and vertex deletions and edge contractions. 
Obviously, the class of cocomparability graphs is not closed 
under taking minors, since that would imply the erroneous 
conclusion that {\em all\/} graphs are 
cocomparability, because all graphs are obtainable from a large 
enough clique 
by taking subgraphs. 

\bigskip 

To see that cocomparability graphs are closed under edge contractions, 
consider a function model, as described above. To contract and edge 
$\{x,y\}$, replace the functions $f_x$ and $f_y$ by one new function 
which swiftly zig-zags between the two functions $f_x$ and $f_y$. 
Clearly, the new function serves as a contraction since any vertex 
intersects $f_x$ or $f_y$ if and only if it intersects the new function 
(see, eg,~\cite{kn:fleischer}). 

\bigskip 

We show that the class of $\AT$-free graphs is closed under edge contractions. 
Actually, even the larger class of `hereditary dominating pair graphs' 
has the property~\cite{kn:przulj}. 

\begin{lemma}
The class of $\AT$-free graphs is closed under edge contractions. 
\end{lemma}
\begin{proof}
Let $G$ be $\AT$-free and let $H$ be obtained from $G$ by contracting 
an edge $\{x,y\} \in E(G)$.  Assume $H$ has an $\AT$, say $\{p,q,r\}$. 
Any path between $p$ and $q$ in $H-N[r]$ corresponds with a path 
in $G$, possibly containing the edge $\{x,y\}$. If $r$ is the contracted 
vertex, 
then the path avoids the neighborhood of both $x$ and $y$ in $G$. 
This shows that $G$ must also contain an $\AT$, which is a contradiction. 
\end{proof}

\bigskip 

Consider contracting every connected, monochromatic subgraph to a single  
vertex. Then, according to 
the lemma above, the resulting graph is still $\AT$-free, 
and, furthermore, the outcome 
is properly colored, that is, each set of colors induces an independent set. 
By the definition of the game, the minimal number 
of moves needed to conquer the graph is the same in both graphs. This 
proves the following corollary. 

\begin{corollary}
The complexity of  
solitaire {\sc Flood-it} on $\AT$-free graphs is equivalent to 
that of the game played on $\AT$-free graphs with a proper vertex coloring, 
that is, where each color class induces an independent set. 
\end{corollary}

\bigskip 

$\AT$-free graphs decompose quite gracefully into blocks and intervals. 

\begin{definition}
Let $G$ be a graph and let $x \in V(G)$. A block at $x$ is a 
component of $G-N[x]$. 
\end{definition}

\begin{definition}
Let $G$ be a graph and let $x$ and $y$ be nonadjacent vertices in $G$. 
A vertex $z$ is between $x$ and $y$ if $x$ and $z$ are contained in a 
common component of $G-N[y]$ and $y$ and $z$ are contained in a common 
component of $G-N[x]$. 
The interval between $x$ and $y$ is the set of all vertices that are 
between $x$ and $y$. 
\end{definition}

\bigskip 

The close relationship between interval graphs 
and $\AT$-free graphs is illustrated by the fact that 
a graph is $\AT$-free if and only if every 
minimal triangulation is an interval graph~\cite{kn:kloks3}. 
This is demonstrated by the following two decomposition theorems, which,   
incidentally, lie at the heart of the 
algorithm that computes the independence number in $\AT$-free graphs~\cite{kn:broersma2}. 

\begin{theorem}
\label{decomposition 1}
Let $G$ be $\AT$-free and let $I(x,y)$ be an nonempty interval 
between two nonadjacent vertices $x$ and $y$. Then, for any vertex 
$z \in I(x,y)$, the removal of $N[z]$ partitions $I(x,y)-N[z]$ 
into two intervals 
$I(x,z)$ and $I(z,y)$ and a collection of blocks at $z$. 
\end{theorem}

\begin{theorem}
\label{decomposition 2}
Let $G$ be $\AT$-free. Let $B$ be a block at a vertex $x$. 
Then, for any vertex $y \in B$, the removal of $N[y]$ partitions 
$B-N[y]$ into 
an interval $I(x,y)$ and some blocks at $y$. 
\end{theorem}

\bigskip 

We end this section with one more definition; that of an extreme. 

\begin{definition}
Let $G$ be a connected graph. A vertex $x$ is an extreme if the largest component 
of $G-N[x]$ has, among all vertices 
of $G$, the maximal cardinality.  In a disconnected graph a vertex is 
extreme if it is extreme in a component of the graph. 
\end{definition}

So, when $G$ is $P_3$-free, that is, 
when $G$ is a clique, or a disjoint union of cliques, 
then every vertex of $G$ is an extreme. 

\bigskip 

\begin{lemma}
\label{lm extreme}
Let $G$ be a connected graph and let 
$x$ be an extreme in $G$ and let $C$ be the largest component of $G-N[x]$. 
Let $S=N(C)$, that is, $S$ is the set of vertices in $V\setminus C$ that 
have a neighbor in $C$. Let 
\[X=V(G) \setminus (N(C) \cup C).\] 
Then $x \in X$ and every vertex of $X$ is adjacent to every vertex of $S$. 
\end{lemma}
\begin{proof}
Assume that some vertex $x^{\prime} \in X$ is not adjacent to 
some vertex $y \in N(C)$. 
This contradicts the assumption that $x$ is extreme, since 
$C \cup \{y\}$ is contained in a component of $G-N[x^{\prime}]$, 
that is, the largest component of $G-N[x^{\prime}]$ is larger than 
the largest component of $G-N[x]$. 
\end{proof}

\begin{corollary}
All vertices of the set $X$ defined as in Lemma~\ref{lm extreme} 
are extreme. 
\end{corollary}

\bigskip 

The set $X$ is a module and, when $G$ is $\AT$-free, it induces an, 
otherwise unrestricted, $\AT$-free 
graph. Anyway, $\AT$-free or not,  
in turn $G[X]$ generally contains an extreme. 

\begin{definition}
A global extreme is, 
\begin{enumerate}
\item an element of $X$ when $G[X]$ is a union of cliques, 
or,  
\item defined recursively, as a global extreme of any component of $G[X]$, otherwise. 
\end{enumerate}
\end{definition}

\section{The case where the source vertex is a global  extreme}

Following the stratagem of Fleischer and Woeginger in~\cite{kn:fleischer} 
we start with an analysis of the case where the source vertex, denoted 
by $x_0$, is a global extreme. Throughout this section we assume that the 
graph $G$ is a connected $\AT$-free graph with a proper vertex coloring. 

\begin{lemma}
\label{lm vertex in interval}
Let $x$ and $y$ be two vertices of the same color 
and assume that $y \in I(x,x_0)$ 
Then conquering $x$ simultaneously conquers $y$. By that we mean 
that when $x$ is added to $x_0$'s territory, either $y$ was already in 
that territory or else it gets added in the same round as $x$. 
\end{lemma}
\begin{proof}
Conquering $x$ implies, by definition, that the territory 
of $x_0$ is colored with color $c(x)$ and that, subsequently,   
there is a $x_0,x$-path with all its 
vertices of  the color $c(x)$. 

\medskip 

\noindent
By Theorem~\ref{decomposition 1}, the neighborhood $N[y]$ separates $x$ and $x_0$ 
into different components of $G-N[y]$.  
Since the set $N[y]$ separates $x_0$ and $x$, 
the $x_0,x$-path goes through a neighbor of $y$. 
Then, calling the color 
$c(x)$ by $x_0$ will add $y$ to the territory 
as well. 
\end{proof}

\bigskip 

\begin{lemma}
\label{lm vertex in block}
Let $x$ and $y$ be two vertices of the same color. 
Assume that $y$ and $x_0$ are in a common block $B$ at $x$. 
Then conquering $x$ simultaneously conquers $y$. 
\end{lemma}
\begin{proof}
By Theorem~\ref{decomposition 2}, $G-N[y]$ partitions the component 
$B$ of $G-N[x]$ into an interval $I(x,y)$ and some blocks at $y$. 
The vertex $x_0 \notin I(x,y)$ since this would contradict that $x_0$ 
is a global extreme. 

\medskip 

\noindent
To see that, assume that $x_0 \in I(x,y)$ and let $C$ be the largest component 
of $G-N[x_0]$. If $x$ and $y$ were both in $C$ then $\{x,y,x_0\}$ 
would be an $\AT$. If $x \in C$ and $y \notin C$, then $x_0$ and 
$y$ cannot be in one component of $G-N[x]$, since $N(C) \subseteq N[y]$ 
by Lemma~\ref{lm extreme}. 
Finally, assume that neither $x$ nor $y$ is in $C$. Then $x_0$ is in 
the interval between $x$ and $y$ in $G-C-N(C)$. However, by induction, this 
is a contradiction, since $x_0$ is a global extreme in $G-C-N(C)$.

\medskip 

\noindent 
Thus,  the closed neighborhood of $y$ separates $x$ and 
$x_0$ or $x_0 \in N[y]$, 
which implies the claim, as in the proof of 
Lemma~\ref{lm vertex in interval}. 
\end{proof}

\bigskip 

\begin{theorem}
\label{thm order color class}
The vertices of a color class can be linearly ordered 
such that, conquering a vertex $x$ simultaneously conquers 
all the vertices in the color class 
that precede $x$ in the ordering.  
\end{theorem}
\begin{proof}
Let $x$ and $y$ be two vertices of the same color. 
We show that the order in which $x$ and $y$ are conquered 
depends on the graph, and not on the conquering strategy
of $x_0$. 

\medskip 

\noindent 
If $y$ is in the block at $x$ that contain $x_0$, 
or if $y$ is between $x$ and $x_0$, the conquering 
of $x$ simultaneously conquers $y$ by 
Lemmas~\ref{lm vertex in interval} and~\ref{lm vertex in block}. 

\medskip 

\noindent 
If one of $x$ and $y$, say $x$, 
is adjacent to $x_0$ then it is easily perceived that  
conquering $y$ simultaneously conquers $x$. 
Assume that $x$ is not adjacent to 
$x_0$ and let $C$ be the component of $G-N[x]$ that contains $x_0$. Denote 
\[S=N(C) \quad \text{then $S \subseteq N(x)$.}\]
We may assume that $y \notin C$. If $y$ is 
not adjacent to all vertices of $S$, then there is 
a $x_0,x$-path that avoids $N[y]$, that is, 
$x$ is in the block at $y$ that contains $x_0$. 

\medskip 

\noindent 
Assume that $y$ is adjacent to all vertices of $S$. 
Then, conquering one of $x$ or $y$ simultaneously 
conquers the other. 
\end{proof}

\bigskip 

\begin{figure}
\setlength{\unitlength}{.1cm}
\begin{center}
\begin{picture}(50,35)
\put(5,20){\circle*{1.3}}
\put(20,10){\circle*{1.3}}
\put(20,30){\circle*{1.3}}
\put(35,10){\circle*{1.3}}
\put(35,30){\circle*{1.3}}
\put(50,20){\circle*{1.3}}
\put(5,20){\line(3,-2){15}}
\put(5,20){\line(3,2){15}}
\put(20,10){\line(0,1){20}}
\put(20,10){\line(1,0){15}}
\put(20,30){\line(1,0){15}}
\put(50,20){\line(-3,-2){15}}
\put(50,20){\line(-3,2){15}}
\put(35,7){$y$}
\put(35,32){$x$}
\put(2,17){$x_0$}
\end{picture}
\end{center}
\caption{This graph contains an $\AT$, namely 
$\{\;x_0,\;x,\;y\;\}$.}
\end{figure}
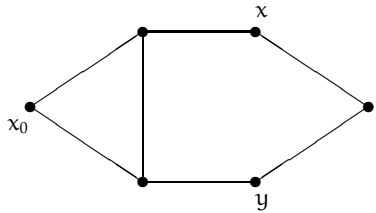

\bigskip 

\begin{theorem}
\label{thm global extreme}
The {\sc solitaire flood-it} game starting at an 
extreme vertex can be solved in 
polynomial time on $\AT$-free graphs. 
\end{theorem}
\begin{proof}
By Theorem~\ref{thm order color class}, each color class 
can be linearly ordered such that conquering a vertex implies 
the conquest of all preceding vertices in the same color class. 
Assume all color classes have been linearly ordered like that and 
denote by $\Max(c)$ the maximal element of color $c$.
Let $C$ be the set of colors and let 
\[M=\{\;\Max(c)\;|\; c \in C\;\}.\] 

\medskip 

\noindent 
Consider a vertex $x$ for which the cardinality of the component 
of $G-N[x]$ that contains the source $x_0$ is as large as possible. 
Let $D$ be the component. Let 
\[\Delta=N(D) \quad \text{and}\quad \Omega=V \setminus (D \cup \Delta).\]
Then $x \in \Omega$ and all vertices of $\Omega$ are adjacent to 
all vertices of $\Delta$. 
Notice that $\{x_0,x\}$ is a dominating pair, that is, each path running 
between $x_0$ and $x$ is a dominating set. 

\medskip 

\noindent
If a color $c$ appears at least once in $\Omega$, then 
$\Max(c) \in \Omega$. 
Let $\Omega^{\ast}=\Omega \cap M$. 

\medskip 

\noindent 
Give vertices of $M$ a weight zero and all other vertices a weight 1. 
Let $\Opt$ denote the cost of a cheapest path to reach at least one 
vertex in $\Omega^{\ast}$.  
Then the optimal solution to solve the game has a cost 
$\Opt +k$, where $k=|C|$ is the number of colors. To see that, consider 
any $x_0,x$-path $P$. Add each vertex of $M$, which is not already 
in the path, to $P$ at a maximal distance 
from $x_0$; this constructs a caterpillar. The strategy which adds the colors 
to the territory of $x_0$ in the order of $P$, adding the vertices of $M$ 
when they are met, eventually adds all vertices to the territory. 
This proves the claim. 
\end{proof}

\section{An algorithm for {\sc Flood-It} on $\AT$-free graphs}

\begin{notation}
Let $\alpha$ and $\omega$ be two vertices such that the interval 
$I(\alpha,\omega)$ contains 
a maximal number of vertices. Let $C_{\alpha}(\omega)$ denote the 
component of $G-N[\alpha]$ that contains $\omega$ and define 
$C_{\omega}(\alpha)$ similarly. Let 
\begin{align*}
S_{\alpha}&=N(C_{\alpha}(\omega)) \quad   
&& A =V(G) \setminus (S_{\alpha} \cup C_{\alpha}(\omega)) & \quad\text{and} \\ 
S_{\omega} &=N(C_{\omega}(\alpha))  \quad   
&& \Omega = V(G) \setminus (S_{\omega} \cup C_{\omega}(\alpha)). &  
\end{align*}
\end{notation}
Notice that, possibly $S_{\alpha} \cap S_{\omega} \neq \es$.

\bigskip 

\begin{lemma}
Under the restrictions set out above, one may choose $\alpha$ such that 
all vertices of $A$ are adjacent to all vertices of $S_{\alpha}$. 
\end{lemma}
\begin{proof}
Assume that there exists a vertex $a \in A$ such that $a$ is not 
adjacent to some vertex $\epsilon \in S_{\alpha}$. Then, the component 
of $G-N[a]$ which contains $\omega$ contains $C_{\alpha}(\omega) \cup \{\epsilon\}$. 
It follows that 
\[I(\alpha,\omega) \subseteq I(a,\omega).\] 
Continuing the process proves the claim. 
\end{proof}
Henceforth, we assume that $\alpha \in A$ 
and that all vertices of $A$ are adjacent to all 
vertices of $S_{\alpha}$ and, similarly, $\omega \in \Omega$ 
and all vertices of $\Omega$ are adjacent to all vertices of $S_{\omega}$. 

\bigskip 

A strategy is a sequence of colors, called by the source vertex $v_0$, which 
ultimately adds all vertices to its territory. 
For the case where $v_0$ is in one of the sets $A$, $\Omega$, 
$S_{\alpha}$ or $S_{\omega}$, the analysis is similar to where $v_0$ is an 
extreme. So, in this section we concentrate on the case where 
\[ \boxed{v_0 \in I(\alpha,\omega).} \] 
Assuming that is the case, $N[v_0]$ separates $\alpha$ and $\omega$ and, 
by Theorem~\ref{thm order color class}, 
each color class is partitioned into two linearly ordered 
sets. In other words, each color 
class has a `maximal' and a `minimal' element. The minimal element 
is the last one visited on the path from $v_0$ that contains some element of $A$ 
and 
the maximal element is the last one visited on the path from $v_0$ 
that contains some element of $\Omega$.  

Notice that $A$ may contain more than one element of some color, but conquering 
one of them, conquers them all. To see that, observe  
that any strategy induces a path from $A$ to 
$\Omega$. The path passes through $S_{\alpha}$ and 
the join between $A$ and $S_{\alpha}$ 
proves the claim. 
Contrary to cocomparability graphs, the sets $A$ and 
$\Omega$ need not be cliques. 

\bigskip 

Analogous the Fleischer \& Woeginger's stratagem for cocomparability graphs, 
we define the essential length of a strategy as follows. 

\begin{definition}
The length of a strategy $\gamma$ is the number of colors in it. 
The essential length of $\gamma$  is the length minus the number of 
steps where the \underline{second} extremal vertex of some color 
class is conquered. 
\end{definition}

For a strategy $\gamma$, let $|\gamma|$ denote its length 
and let $\ess(\gamma)$ denote its essential length. 
Let $\Opt$ denote the solution of the solitair {\sc Flood-It} game, 
then we have 
\[\Opt=\min \;\{\; \ess(\gamma)+k\;|\; \text{$\gamma$ is a strategy}\;\},\]
where $k$ is the number of colors. 

\bigskip 

\begin{notation}
For a vertex $x$, let $\Min_x(c)$ denote the minimal element of color $c$ 
which is adjacent to $x$ and let $\Max_x(c)$ denote the maximal element of color 
$c$ which is adjacent to $x$. 
\end{notation}

\begin{notation}
For a pair of vertices $x$ and $y$ let $D(x,y)$ denote the essential length 
of a strategy that conquers $x$ and $y$. We set $D(v_0,v_0)=0$. 
\end{notation}

\begin{definition}
Two vertices $x$ and $x^{\prime}$ that are not adjacent to $v_0$
are incomparable  
if  
the component of $G-N[x^{\prime}]$ that contains $v_0$ is the same 
as the component of $G-N[x]$ that contains $v_0$. 
\end{definition}

\begin{theorem}
Let $x$ be a vertex in a component of $G-N[v_0]$ which contains some vertex of $A$ 
and let $y$ be a vertex in a component of $G-N[v_0]$ which contains a vertex of $\Omega$. 
Then 
\[D(x,y)=\min\; \{\; D(x,\Min_y(c))+\delta_y(x),\; D(\Max_x(c),y)+\delta_x(y) \;|\; 
\text{$c$ is a color}\;\},\]
where $\delta_x(y)=0$ if $y$ is the maximal element of some color $c$ and the 
minimal element is not in $I(x,y)$, or adjacent to $x$, or incomparable to $x$ or $y$ 
or both and
$\delta_x(y)=1$ otherwise. In other words, $\delta_x(y)=0$ if $y$ is a maximal 
element of a color class and the minimal element is conquered earlier 
or in the same step. 
\end{theorem}
As usual, we let $\infty$ be the minimal value of a set which is empty, 
so, for example, if $y$ is not adjacent to a vertex of color $c$, then 
we let $D(x,\Min_y(c))=\infty$. 
\begin{proof}
In analogy to the proof in~\cite{kn:fleischer}, let $\gamma$ be an optimal 
strategy and let $\alpha \in A$ and $\omega \in \Omega$ be the first two 
extremal vertices that are conquered. 
After the conquest of $\alpha$ and $\omega$, only the remaining maximal elements 
need to be conquered. Thus $|\gamma|\geq D(\alpha,\omega)+k$. 
\end{proof}

\section{Concluding remarks}

We have shown that the solitaire {\sc Flood-It} game can be solved in 
polynomial time on $\AT$-free graphs.  A graph is a hereditary dominating 
pair graph if each of its connected induced subgraphs has a dominating 
pair~\cite{kn:przulj}. 
$\AT$-free graphs are hereditary dominating pair graph. That the latter is 
actually a larger class of graphs is exemplified by $C_6$. 
As far as we know, the recognition of hereditary 
dominating pair graphs is still open.  An interesting open question is whether 
solitaire {\sc Flood-It} remains polynomial for this class of graphs.


\begin{thebibliography}{99}

\bibitem{kn:backer}Backer,~J., 
Separator orders in interval, cocomparability and $\AT$-free graphs, 
{\em Discrete Applied Mathematics\/} {\bf 159} (2011), pp.~717--726. 

\bibitem{kn:ben-ameur}Ben-Ameur,~W., M.~Mohamed-Sidi and J.~Neto, 
The $k$-separator problem: polyhedra, complexity and approximation results, 
Journal of Combinatorial Optimization {\bf 29} (2015), pp.~276--307. 

\bibitem{kn:bhowmick}Bhowmick,~D. and L.~Chandran, 
Boxicity and cubicity of asteroidal triple free graphs, 
{\em Discrete Mathematics\/} {\bf 310} (2010), pp.~1536--1543. 

\bibitem{kn:broersma2}Broersma,~H., T.~Kloks, D.~Kratsch and 
H.~M\"uller, 
Independent sets in asteroidal triple-free graphs, 
{\em SIAM Journal on Discrete Mathematics\/} {\bf 12} (1999), 
pp.~276--287. 

\bibitem{kn:broersma}Broersma,~H., T.~Kloks, D.~Kratsch and H.~M\"uller, 
A generalization of $\AT$-free graphs and a generic algorithm for 
solving triangulation problems, 
{\em Algorithmica\/} {\bf 32} (2002), pp.~594--610. 

\bibitem{kn:clifford}R.~Clifford, M.~Jalsenius, A.~Montanaro 
and B.~Sach, 
The complexity of flood filling games. 
Manuscript on ArXiv~1001.4420, 2010. 

\bibitem{kn:corneil}Corneil,~D. and J.~Stacho, 
Vertex ordering characterizations of graphs of bounded 
asteroidal number, 
Journal of Graph Theory~{\bf 78} (2014), pp.~61--79. 

\bibitem {kn:fleischer}
R.~Fleischer and G.~Woeginger, 
An algorithmic analysis of the Hony-Bee game, 
{\em Theoretical Computer Science\/} {\bf 452} (2012), 
pp.~75--87. 

\bibitem{kn:fukui}Fukui,~H., Y.~Otachi, R.~Uehara, T.~Uno and 
Y.~Uno, 
On the complexity of flooding games on graphs with 
interval representations. Manuscript on arXiv: 1206.6201, 2012. 

\bibitem{kn:golumbic}Golumbic,~M., D.~Rotem and J.~Urrutia, 
Comparability graphs and intersection graphs, 
{\em Discrete Mathematics\/} {\bf 43} (1983), pp.~37--46. 

\bibitem{kn:halin}Halin,~R., 
Some remarks on interval graphs, 
{\em Combinatorica\/} {\bf 2} (1982) pp.~297--304. 

\bibitem{kn:kloks3}Kloks,~T., 
{\em Treewidth -- Computations and Approximations\/}, Springer-Verlag, 
Berlin, Lecture Notes in Computer Science 842, 1994. 

\bibitem{kn:kloks2}Kloks,~T., D.~Kratsch and H.~M\"uller, 
Asteroidal sets in graphs, 
{\em Proceedings $19^{\mathrm{th}}$ WG'97\/}, Springer-Verlag, Berlin, 
Lecture Notes in Computer Science LNCS~1335 (1997), pp.~229--241. 

\bibitem{kn:kloks4}Kloks,~T., H.~M\"uller and C.~Wong, 
Vertex ranking of asteroidal triple-free graphs, 
{\em Information Processing Letters\/} {\bf 68} (1998), pp.~201--206.

\bibitem{kn:kloks}Kloks,~T. and Y.~Wang, 
Advances in graph algorithms. Manuscript on viXra: 1409.0165, 2014. 

\bibitem{kn:kohler}K\"ohler,~E., 
Recognizing graphs without asteroidal triples, 
{\em Journal of Discrete Algorithms\/} {\bf 2} (2004), pp.~439--452. 

\bibitem{ln:lagoutte}Lagoutte,~A., M.~Noual and E.~Thierry, 
Flooding games on graphs. Manuscript HAL-00653714, 2011. 

\bibitem{kn:lekkerkerker}Lekkerkerker,~C. and D.~Boland, 
Representation of finite graphs by a set of intervals on the real line, 
{\em Fundamenta Mathematicae\/} {\bf 51} (1962), pp.~45--64. 

\bibitem{kn:meeks2}Meeks,~K. and A.~Scott, 
The complexity of Free-Flood-It on $2 \times n$ boards. 
Manuscript on arXiv: 1101.5518, 2011. 

\bibitem{kn:meeks}~Meeks,~K. and A.~Scott, 
Spanning trees and the complexity of flood filling games. 
Manuscript on arXiv: 1203.2538, 2012. 

\bibitem{kn:przulj}Pr\u{z}ulj,~N., D.~Corneil and E.~K\"ohler, 
Hereditary dominating pair graphs, 
{\em Discrete Applied Mathematics\/} {\bf 134} (2004), pp.~239--261. 

\bibitem{souza}Souza, ~U., F.~Protti and M.~da~Silva, 
An algorithmic analysis of Flood-it and Free-Flood-it on graph powers, 
{\em Discrete Mathematics and Theoretical Computer Science\/} {\bf 16} (2014), 
pp.~279--290. 

\bibitem{kn:souza2}Souza,~U., F.~Protti, M.~da~Silva, 
Parameterized complexity of flood-filling games on trees, 
{\em Proceedings COCOON~2013\/}, Springer-Verlag, LNCS~7936 
(2013), pp.~531--542. 

\bibitem{kn:walter}Walter,J., 
Representations of chordal graphs as subtrees of a tree, 
{\em Journal of graph theory\/} {\bf 2} (1978), pp.~265--267. 

\end{thebibliography}
\end{document}